%% file: main.tex
\documentclass[letterpaper, 10 pt, conference]{ieeeconf}  

\IEEEoverridecommandlockouts                              
\usepackage[utf8]{inputenc}
\usepackage[sorting=none,maxbibnames=99,giveninits]{biblatex} 
\usepackage{graphicx} 
\usepackage{amsmath}
\usepackage{amsthm}
\usepackage{amssymb}

\newtheorem{Theorem}{Theorem}
\newtheorem{definition}{Definition}

\usepackage{mathtools}

\newcommand\norm[1]{\left\lVert#1\right\rVert}
\newcommand\abs[1]{\lvert#1\rvert}
\newcommand{\defeq}{\vcentcolon=}
\usepackage{multirow}
\usepackage{hhline}
\usepackage{amsfonts} 

\usepackage{accents}
\newcommand{\ubar}[1]{\underaccent{\bar}{#1}}
\usepackage{algorithm}
\usepackage{algpseudocode}

\usepackage[belowskip=-10pt,aboveskip=2pt]{caption}
\title{\LARGE \bf
Synthesis and verification of robust-adaptive safe controllers 
}

\author{Simin Liu$^{*1}$ \thanks{$^*$ Denotes equal contribution} \thanks{$^1$Robotics Institute, Carnegie Mellon University, Pittsburgh, PA 15213, USA.}  \and Kai S. Yun$^{*2}$ \thanks{$^2$Mechanical Engineering Department, Carnegie Mellon University, Pittsburgh, PA 15213, USA.} \and John M. Dolan$^1$ \and Changliu Liu$^1$}

\begin{document}

\maketitle
\thispagestyle{empty}
\pagestyle{empty}

\begin{abstract}
\input{abstract}
\end{abstract}

\input{intro}
\input{related_works}

\input{prob_formulation}

\input{prelims}

\input{method}
\input{results}

\input{conclusion}





\section*{Acknowledgements}

The authors would like to thank Hongkai Dai at Toyota Research Institute and Weiye Zhao at Carnegie Mellon University for their technical guidance and paper revisions. This work was supported by the Qualcomm Innovation Fellowship awarded to Simin Liu and National Science Foundation under Grant No. 2144489.



\AtNextBibliography{\small}

\printbibliography


\section{Appendix}\label{sec:appendix}

\input{appendix}

\end{document}

%% file: abstract.tex
Safe control with guarantees generally requires the system model to be known. It is far more challenging to handle systems with uncertain parameters. In this paper, we propose a generic algorithm that can synthesize and verify safe controllers for systems with constant, unknown parameters. In particular, we use robust-adaptive control barrier functions (raCBFs) to achieve safety. We develop new theories and techniques using sum-of-squares that enable us to pose synthesis and verification as a series of convex optimization problems. In our experiments, we show that our algorithms are general and scalable, applying them to three different polynomial systems of up to moderate size (7D). Our raCBFs are currently the most effective way to guarantee safety for uncertain systems, achieving 100\% safety and up to 55\% performance improvement over a robust baseline.

%% file: intro.tex
\section{Introduction}\label{sec:intro}

Although much work has been done on safety for known systems, comparatively little work has been done for systems whose dynamical models cannot perfectly be known. These uncertain systems are a challenging domain for safe control on both a theoretical and methodological level. Effective safe controllers must meet two demands: (1) provide mathematical guarantees of safety and (2) be high-performance, interfering minimally with other control objectives (i.e., stabilizing, tracking). Currently, many safe control synthesis methods for uncertain systems either lack guarantees or only obtain safety at a high cost to performance. Secondly, algorithms for controller search should be both generic and tractable, so they can handle many different systems of varying dimensions. This work aims to address all of these requirements. 


The technique we use is “sum-of-squares programming” (SOSP)~\cite{papachristodoulou2005tutorial,jarvis2003some}. Any robust-adaptive CBF (raCBF) produced by our SOSP-based algorithm is guaranteed to provide safety. SOSP allows us to pose the search for a valid raCBF, given a partially-known system and a safety requirement, as a constrained optimization problem. This optimization problem can be broken down into a series of convex optimization problems, which can each be solved with an off-the-shelf semidefinite program (SDP) solver. 


\begin{figure}[htbp]
    \centering
    \includegraphics[width=0.45\textwidth]{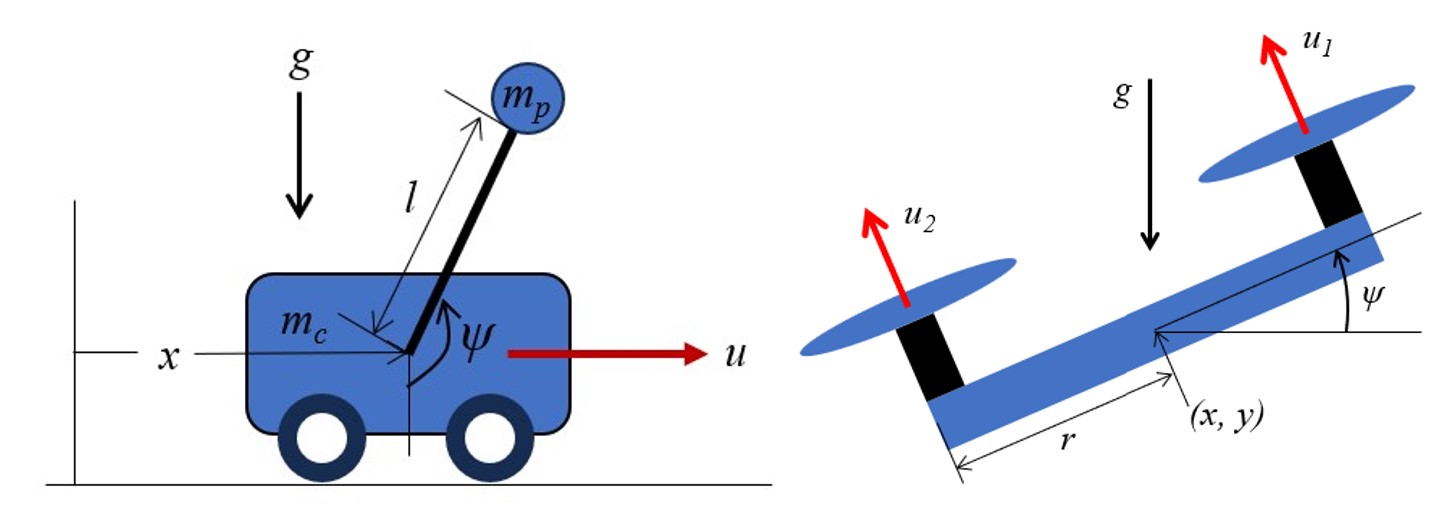}
    \caption{Two of our three test systems: cartpole with unknown joint friction and quadrotor with unknown drag coefficients.}\label{fig:systems_side_by_side}
\end{figure}

Typically, uncertain systems are handled with adaptive, robust, or robust-adaptive control~\cite{lavretsky2012robustAdaptive}. We choose the robust-adaptive paradigm, which combines the merits of robust and adaptive control. In adaptive control, the controller depends on estimates of the unknown system. Typically, the system will be estimated online and the controller {\it adapts} with each new estimate. In adaptive fashion, our raCBF is a function of the estimated parameters as well as the state. This raCBF is also paired with a special parameter estimation law. This controller structure allows it to effectively {\it select} a CBF online with the choice of parameter. The robust-adaptive scheme leverages this to attain {\it selective conservatism}: more conservative CBFs are only ``selected'' when necessary, like when the state suddenly moves toward danger.     

The raCBF also incorporates elements of robust control. In robust control, the controller only depends on worst-case disturbance bounds. Our raCBF follows robust principles by accounting for maximum estimation error. However, it is able to account for it in a less conservative way compared to robust CBF (rCBF), as it assumes parameter estimation occurs. Ultimately, these differences lead our synthesized raCBF to perform better. That is, it provides safety while minimally hindering control efficiency at tasks like stabilization and tracking. 

Our contributions are as follows:
\begin{itemize}
\item Propose the first verification and synthesis methods for adaptive-type safe controllers for uncertain systems.
\item For verification: derive convex-equivalent conditions for an raCBF to provide safety.
\item For synthesis: design a multi-phase algorithm which yields a valid raCBF with a locally optimal invariant set.
\item Illustrate our algorithm's genericness and scalability by applying it to three examples with varying dynamics, dimensions ($\leq$ 7D), and safety specifications.
\item Illustrate our synthesized raCBF's safety and superior  performance (up to $55\%$ improvement over robust baseline) in simulation.
\end{itemize}

%% file: related_works.tex
\section{Related works}\label{sec:related_works}

In this section, we briefly survey existing works on safety for uncertain systems. We primarily focus on CBF works.


{\bf Theory of (r/a)-CBFs:} these works mathematically define novel forms of CBF (like robust~\cite{xu2015rCBF, jankovic2018robust, breeden2021robust} and adaptive~\cite{taylor2020adaptive, maghenem2021adaptiveSafety, isaly2021multipleBF} CBFs). This category has a different focus than synthesis and verification (systematically generating CBFs). Hence, works in this category typically just use hand-derivation to find CBFs for toy systems, which is an approach that is difficult, if not impossible, to apply to more complex systems. 


{\bf Verification and synthesis of (r/a)-CBFs:} this category, which our work belongs to, deals with developing new theories and algorithms for systematically checking and generating safety functions. While synthesis of ordinary CBFs (i.e., for known systems) has been explored extensively~\cite{zhao2022safety, clark2022semi, dai2022convex}, to our knowledge, there are very few works for robust CBF synthesis~\cite{kang2023verification} and none for adaptive synthesis. Our work is one of a handful that synthesize CBFs for uncertain systems, and the first to synthesize an adaptive-type CBF.   
While our synthesis technique carries guarantees of correctness, there are ad-hoc synthesis works which do not~\cite{liu2023safe, wei2022safe}. There are also sample-based methods for verification, which are unable to scale past 5D due to the curse of dimensionality~\cite{wei2022safe}. Compare this to our SOSP-based verification scheme, which has polynomial time complexity and can scale to 7D and possibly beyond. 



{\bf Online optimization of (r/a)-CBFs:} this body of work typically assumes a valid CBF already exists. Ordinary CBFs compute safe control inputs using quadratic programs (QPs) solved at each time step~\cite{ames2019control, wei2019safe}. Works in this category investigate how best to pose this online optimization (i.e., possibly as a different type of program, like a second-order cone program) for new forms of CBFs~\cite{wei2022persistently, hamdipoor2023safe, long2022safe}. Note that one advantage of our robust-adaptive formulation is that it can still be applied with a QP, which is simple and fast to solve.   


{\bf Ad-hoc safe controllers: } these are works which cannot guarantee safety, because they cannot prove that an admissible safe control input exists for all states and all realizations of the unknown system~\cite{xiao2021adaptive, parwana2022recursive, lyu2021merge}. By contrast, our approach can ensure that there exists a safe control input for {\it all possible conditions}, therefore guaranteeing system safety. 

{\bf Non-CBF approaches: } there are bodies of work on using HJB reachability~\cite{bansal2017hamilton} or set-based methods~\cite{althoff2010computing, althoff2013reachability} to compute robust backward or forward reachable sets. However, a significant drawback of these methods is they scale poorly, often not beyond 5D. 




%% file: prob_formulation.tex
\section{Problem formulation}\label{sec:prob_formulation}
\subsection{Safe control for uncertain systems}
In this first subsection, we define the problem of safe control for uncertain systems. We also outline system assumptions and our definition of safety. 
We consider a system that is affine in the unknown parameters and the control input: 
\begin{align}
    \dot{x} = f(x) - \Delta(x) \theta + g(x) u \label{eqn:sys_dynamics}\\
    u \in \mathcal{U}, x \in \mathcal{X}, \theta \in \Theta
\end{align}
where $f, g, \Delta$ are polynomial functions in $x$ and $x \in \mathcal{X} \subseteq \mathbb{R}^{n}, u \in \mathcal{U} \subseteq \mathbb{R}^{m}, \theta \in \Theta \subset \mathbb{R}^{p}$. We also assume the set of admissible inputs and the robust parameter range are both hyperrectangles (i.e., $ \ubar{u} \leq u \leq \bar{u}$ and $\ubar{\theta} \leq \theta \leq \bar{\theta}$). We want to enforce safety of this system in the sense of “set invariance”. 

\begin{definition}[Set invariance]
A set $\mathcal{I}$ is forward invariant if all trajectories starting inside $\mathcal{I}$ remain inside the set for all time.
\end{definition}
The set $\mathcal{X}_{\text{safe}}$ is provided as part of the problem specification and is represented implicitly using continuous and smooth functions $b_{i}(x) : \mathbb{R}^{n} \to \mathbb{R}$: 
\begin{definition}[Safety specification]
\begin{equation}
     \mathcal{X}_{\text{safe}} = \{x | b_{i}(x) \geq 0, \; \forall i=1,\ldots,t\} \label{eqn:sys_safety_def}
\end{equation}
\end{definition}

\subsection{Robust-adaptive CBFs}
In this second subsection, we summarize the raCBF formulation and outline the constraints that valid raCBFs must satisfy. We choose the formulation in~\cite{lopez2023unmatched}, which is less conservative than~\cite{xu2015rCBF, jankovic2018robust, breeden2021robust, taylor2020adaptive, isaly2021multipleBF} and compatible with more nominal controllers than~\cite{maghenem2021adaptiveSafety}.  
The raCBF has the form $\phi(x, \hat{\theta})$ where it depends on not only the state, but also on the estimate of the unknown parameters. The estimate, $\hat{\theta}$, is initialized with a guess and then updated online according to:
{\small
\begin{align}
&\dot{\hat{\theta}} = \gamma \nu(\rho) \Delta(x)^{\top} \nabla_{x} \phi(x, \hat{\theta}), \\
&\dot{\rho} = \frac{\nu(\rho)}{\nabla \nu(\rho)} \frac{1}{\phi(x, \hat{\theta}) + \eta}  \left[ -\sigma \rho + w_{\hat\theta}(x) \right], \\
&w_{\hat\theta} = \begin{cases} 0 \text{ if } \nabla_{\hat{\theta}} \phi(x, \hat{\theta})^{\top} \dot{\hat{\theta}} \ge 0\\
-\zeta \nabla_{\hat{\theta}} \phi(x, \hat{\theta})^{\top}  [\gamma \Delta(x)^{\top} \nabla_x \phi(x,\hat{\theta})] \text{ otherwise}
\end{cases}
\end{align}
}where $\gamma$ is an admissible adaptation gain, $\nu(\rho)$ is a scaling function such that $1 \leq \nu(\rho) \leq \zeta < \infty$ for $\zeta > 1$ and also $\nabla \nu(\rho) > 0$. Also, $\eta, \sigma, \zeta \in \mathbb{R}_{>0}$ are design parameters. We assume that $\theta$ belongs to a bounded set $\Theta$. For many systems, we can easily provide a guess for $\Theta$ based on domain knowledge~\cite{lavretsky2012robustAdaptive}. 





\begin{Theorem}\label{thm:valid_raCBF}
Define the set {\small $\mathcal{I}^{r}_{\hat{\theta}} = \{x \in \mathcal{X}, \hat\theta\in\Theta \;|\; \phi(x, \hat{\theta}) \geq \frac{1}{2 \gamma} \tilde{\theta}^{\top} \tilde{\theta} \}$} where $\tilde{\theta}$ is the maximum possible estimation error. {\small $\tilde{\theta} = \bar{\theta} - \ubar{\theta}$}, where {\small $\theta \in [\ubar{\theta}: \bar{\theta}]$}. Then, an raCBF must satisfy the following constraints to ensure invariance of $\mathcal{I}^{r}_{\hat{\theta}}$:
{\small
\begin{align}
    \sup_{u \in \mathcal{U}} \left\{ \nabla_{x} \phi(x, \hat{\theta}) (f - \Delta \hat{\theta} + gu ) \right\} \geq& -\alpha \big( \phi(x, \hat{\theta}) - \tfrac{1}{2\gamma}\tilde{\theta}^{\top}\tilde{\theta} \big) \label{eqn:ctrl_fltr_feas_constraint}\\
    \phi(x, \hat{\theta}) \geq 0 \implies b&_{i}(x) \geq 0, \forall i = 1,\ldots,t. 
    \label{eqn:subset_constraint}
\end{align}
}
\end{Theorem}
\begin{proof}
    Eq.~(\ref{eqn:ctrl_fltr_feas_constraint}) by Thm. 1 in~\cite{lopez2023unmatched} and Eq.~(\ref{eqn:subset_constraint}) by defintion of safety. 
\end{proof}

Eq.~(\ref{eqn:ctrl_fltr_feas_constraint}) says that the raCBF must be a function that defines a robust invariant set under the true dynamics. Specifically, the robust set $\mathcal{I}^{r}_{\hat{\theta}}$ is invariant iff there always exists an admissible input to repel the state from the set boundary. This is implied by Eq.~(\ref{eqn:ctrl_fltr_feas_constraint}), which requires {\small $\sup_{u \in \mathcal{U}} \dot{\phi} \geq 0$} (repulsion) when {\small $\phi = \frac{1}{2 \gamma} \tilde{\theta}^{\top} \tilde{\theta}$} (state at the boundary). Now, typically, the invariant set is considered $\{x | \phi \geq 0 \}$. The robust set is shrunk from this typical set by a margin proportional to the maximum estimation error. Intuitively, this triggers safe control earlier than is typical, which mitigates the impact of estimation error. 

Additionally, the raCBF must obey the safety specification. Eq.~(\ref{eqn:subset_constraint}) says that the raCBF may define an invariant set that is no larger than $\mathcal{X}_{\text{safe}}$. Often, $\mathcal{X}_{\text{safe}}$ itself cannot be an invariant set. There will be some states within $\mathcal{X}_{\text{safe}}$ that will leave the set for any admissible control input. 




Given a valid raCBF, a controller design is immediate. We simply apply the raCBF as an optimization-based safety-filter on top of a nominal control policy. Since any control input satisfying  Eq.~(\ref{eqn:ctrl_fltr_feas_constraint}) will keep the system safe, this filter simply projects the nominal control signal to the space of $u$ satisfying Eq.~(\ref{eqn:ctrl_fltr_feas_constraint}). 

\begin{definition}[raCBF quadratic program (raCBF-QP)]
Given nominal controller $u_{\text{ref}} = \pi(x)$,
\begin{align}
    & \min_{u \in \mathcal{U}} \norm{u_{\text{ref}} - u}_2^2 \\
      \nabla_{x} \phi(x, \hat{\theta}) (f - \Delta &\hat{\theta} + gu )  \geq  -\alpha \left( \phi(x, \hat{\theta}) - \tfrac{1}{2\gamma}\tilde{\theta}^{\top}\tilde{\theta} \right). 
\end{align}
\end{definition}


Notice that the filtered input is computed using the {\it estimated} dynamics. This is why we require the raCBF, and the filtered input, by consequence, to be robust to estimation error.

%% file: prelims.tex
\section{Preliminaries: SOSP}
We use sum-of-squares programming (SOSP) as the optimization framework for verification and synthesis. In the following section, we briefly summarize key definitions and theorems for SOSP. In general, checking polynomial non-negativity is NP-hard. However, checking that some $p(x) \in \mathcal{R}_n$ (the set of polynomials in $n$ variables) is an SOS polynomial (Def.~\ref{defn:sospoly}) is completely tractable~\cite{parrilo2000structured}. 
\begin{definition}[Sum-of-squares polynomial]
\label{defn:sospoly}
A polynomial $p(x)$ is SOS if $p(x) = \sum_{i} q_{i}(x)^2$ for some polynomials $q_{i}$. Clearly, this implies $p(x) \geq 0$. We denote the set of SOSP as $\Sigma$. 
\end{definition}

Hence, a popular kind of optimization problem is the SOSP, which has constraints of this form (enforcing a polynomial to be sum-of-squares). 

\begin{definition}[Sum-of-squares program]
An SOSP is a convex optimization problem of the following form for a given cost vector $c \in \mathbb{R}^{p \times q}$
\begin{align}
    \min_{w}\;\; &c^{\top} w \label{eqn:sosp_objective}\\
    &a^{i}_{0}(x) + \sum_{\forall j \in 1:p} w^{i}_{j} a^{i}_{j}(x)  \in \Sigma \;\;, \; \forall i =1,\ldots, q\label{eqn:sosp_constraints}
\end{align}
where $a^{i}_{j}(x) \in \mathcal{R}_n$.
\end{definition}

The definition above says that an SOSP constrains $q$ polynomials to be sum-of-squares. These polynomials may have variable coefficients ($w^i_j$) and the objective is linear in these coefficients. Next, we introduce the cornerstone theorem, the Positivstellensatz (P-satz), which allows us to transform all different kinds of constraints (e.g., logical implications of polynomial inequalities) into the standard SOSP constraint form of Eq.~(\ref{eqn:sosp_constraints}). 

\begin{Theorem}[Positivstellensatz~\cite{stengle1974nullstellensatz}]
Given polynomials $\{f_1,\ldots,f_r\}$, $\{g_1,\ldots,g_t\}$, and $\{h_1,\ldots,h_u\}$ in $\mathcal{R}_n$, the following are equivalent: \\
(a) The set 
    \begin{align*}
    \left\{
    x \in \mathbb{R}^{n}
    \middle| \; 
            \begin{aligned}
            & f_{1}(x) \geq 0, \ldots, f_{r}(x) \geq 0\\      
            & g_{1}(x) \neq 0, \ldots, g_{t}(x) \neq 0 \\
            & h_{1}(x) = 0, \ldots, h_{u}(x) = 0
            \end{aligned}
    \right\} 
    \end{align*}
is empty. \\
(b) There exist polynomials $f \in \mathcal{P}(f_1,\ldots,f_r), g \in \mathcal{M}(g_1,\ldots,g_t), h \in \mathcal{I}(h_1,\ldots,h_u)$ such that 
\begin{equation}
    f + g^2 + h = 0
\end{equation}
where $\mathcal{M}(g_1,\ldots,g_t)$ is the multiplicative monoid generated by the $g_i$'s and consists of all their finite products; $\mathcal{P}(f_1,\ldots,f_r)$ is the cone generated by the $f_i$'s:
\begin{align}
    \mathcal{P}(f_1,\ldots,f_r) = 
\left\{  s_0 + \sum_{i=1}^{l} s_i b_i \middle\vert \begin{array}{l} l \in \mathbb{Z}_{+}, s_i \in \Sigma_n,  \nonumber \\  b_i \in \mathcal{M}(f_1,\ldots,f_r) \end{array} \right\}
\end{align}
and finally, $\mathcal{I}(h_1,\ldots,h_u)$ is the ideal generated by the $h_i$'s: 
\begin{equation}
\mathcal{I}(h_1,\ldots,h_u) = \left\{ \sum h_k p_k \middle| p_k \in \mathcal{R}_n \right\}.
\end{equation}
\end{Theorem}
  
The P-satz allows us to convert statements in the form of (a) to the form of (b), which can then straightforwardly be turned into an SOSP-type constraint. However, it is often inconvenient to apply P-satz directly, since it generates complex constraints (e.g., the cone adds $2^r$ terms, each with a polynomial multiplier). Instead, we can use a simplification of P-satz called the S-procedure: 

\begin{definition}[S-procedure~\cite{parrilo2000structured}]
\label{def:s_procedure}
A sufficient condition for proving $p(x) \geq 0$ on the semi-algebraic set $K = \{x \in \mathbb{R}^{n} | b_{1}(x) \geq 0, \ldots, b_{m}(x) \geq 0\}$ is the existence of polynomials $r_{i}(x), i = 0, \ldots, m$ that satisfy 
\begin{align}
    (1 + r_{0}(x)) p(x) - \sum_{i=1}^{m} r_{i}(x) b_{i}(x) \in \Sigma \\
    r_{i}(x) \in \Sigma, i=0,\ldots,m.
\end{align}
\end{definition}


Frequently, we simplify this condition further by taking $r_{0}(x) = 0$. Together, the concepts from this section allow us to write the constraints which define valid raCBFs in the form of SOSP constraints, which makes them amenable to verification and synthesis.

%% file: method.tex
\section{Methodology}\label{sec:methodology}
In the following sections, we show how a valid raCBF can be characterized using several SOSP constraints. With this, we can pose raCBF verification as an SOSP; this allows us to efficiently certify whether a given raCBF is valid. Additionally, we propose an algorithm for synthesizing an raCBF, given the partially unknown system and the safety definition. The algorithm is a sequence of SOSPs that are iteratively solved until convergence.  

\subsection{raCBF verification}\label{sec:methodology:verification}
In this section, we derive an SOSP that allows us to verify if a given raCBF satisfies Theorem~\ref{thm:valid_raCBF}. Our goal is to convert Eq.~(\ref{eqn:ctrl_fltr_feas_constraint}),~(\ref{eqn:subset_constraint}) into the form of Eq.~(\ref{eqn:sosp_constraints}).

Looking first at Eq.~(\ref{eqn:ctrl_fltr_feas_constraint}), we see that it requires us to show that for all $x \in \mathcal{X}, \theta \in \Theta$, there exists some $u \in \mathcal{U}$ satisfying the constraint. The ``exists'' quantifier complicates the conversion, as theorems like P-satz and S-procedure can only support “for all” quantifiers. However, it turns out that for any $x$, we can identify which $u \in \mathcal{U}$ achieves the supremum in Eq.~(\ref{eqn:ctrl_fltr_feas_constraint}) - call it $u^{*}(x)$. Therefore, we only need to check whether Eq.~(\ref{eqn:ctrl_fltr_feas_constraint}) is satisfied for $u^{*}(x)$ for all $x$. 

We derive $u^{*}(x)$ by regarding Eq.~(\ref{eqn:ctrl_fltr_feas_constraint}) as a constrained optimization over the input $u$. In fact, it is a special kind of optimization problem with an affine objective and a compact, convex constraint set. Hence, the maximizer always occurs at a vertex of the constraint set~\cite{dantzig2002linear}, $v^j$ where  $v_i^j \in \{\ubar{u}_i, \bar{u}_i\}$. For each $x$, we can further identify which vertex is the maximizer. First, we write Eq.~(\ref{eqn:ctrl_fltr_feas_constraint}) more clearly in control-affine form. Let {\small $c(x, \hat{\theta}) \defeq \nabla_{x} \phi(x, \hat{\theta}) (f - \Delta \hat{\theta} ) + \alpha ( \phi(x, \hat{\theta}) - \tfrac{1}{2\gamma}\tilde{\theta}^{\top}\tilde{\theta} )$} and {\small $d(x, \hat{\theta}) \defeq \nabla_{x} \phi(x, \hat{\theta}) g(x)$}, so that Eq.~(\ref{eqn:ctrl_fltr_feas_constraint}) can be rewritten as {\small $\sup_{u \in \mathcal{U}} c(x, \hat{\theta}) + d(x, \hat{\theta}) u$}. For $i = 1, \ldots, m$, if $d(x, \theta)_i \geq 0$, i.e., the $i$\textsuperscript{th} entry of vector $d(x, \theta)$ is nonnegative, then $u^{*}_{i}(x) = \bar{u}_i$ and vice versa. There are then $2^m$ groups of $x$ with distinct sign patterns on $d(x, \theta)$ that have distinct maximizers $u^{*}(x)$. Hence, we can write $2^m$ constraints for all the groups of $x$.

We give an example for $m = 1$, although the general case is an immediate extension. We have two constraints:
\begin{align}
    & \forall x \in \mathcal{X}, \theta \in \Theta, \;\;\text{s.t.}\;\; \phi(x, \theta) \geq \beta^{-}: \nonumber \\
    & d(x, \theta) \geq 0 \implies c(x, \theta) + d(x, \theta) \bar{u} \geq 0, \\
    & d(x, \theta) \leq 0 \implies c(x, \theta) + d(x, \theta) \ubar{u} \geq 0.
\end{align}

Applying the S-procedure gives us the following SOSP constraints: 
\begin{align}
    c + d \bar{u} - s_1 d-s_2(\phi - \beta^{-}) &\in \Sigma, \;\; s_1, s_2 \in \Sigma, \label{eqn:sosp_constraint_u1} \\ 
    c + d \ubar{u} + s_3 d-s_4(\phi - \beta^{-}) &\in \Sigma, \;\; s_3, s_4 \in \Sigma. \label{eqn:sosp_constraint_u2} 
\end{align}

Likewise, we can also apply the S-procedure to \eqref{eqn:subset_constraint} to yield:
\begin{align}
    \phi(x, \theta) - r_i b_i \in \Sigma ,\;\; \forall i=1,\ldots,t \;,\; r_i \in \Sigma. \label{eqn:sosp_constraint_subset}
\end{align}


\begin{Theorem}[raCBF verification]
For the case where $m=1$\footnote{The case for general $m$ is an immediate extension, but omitted for notational brevity.}, a polynomial $\phi(x, \hat{\theta})$ is a valid raCBF, providing set invariance and finite-time convergence in $\{x | \phi \geq \beta^{-}\}$, if there exists $s_{1:4}, r_i$ such that the convex constraints in Eq.~(\ref{eqn:ctrl_fltr_feas_constraint}) and~(\ref{eqn:subset_constraint}) are satisfied. This can be determined by seeing if the SOSP consisting of constraints~(\ref{eqn:sosp_constraint_u1}),~(\ref{eqn:sosp_constraint_u2}),~(\ref{eqn:sosp_constraint_subset}) is feasible.         
\end{Theorem}
\begin{proof}
If there exists such $s_{1:4}, r_i$, then by Def.~\ref{def:s_procedure}, Thm.~\ref{thm:valid_raCBF} is satisfied. 
\end{proof}

\begin{algorithm}[ht]
\caption{Synthesis of raCBFs via a sequence of three SOSPs} \label{alg:raCBF}
\begin{algorithmic}

\State Start with {\small $\phi^{(0)}, \theta_l^{(0)}, \theta_u^{(0)}, c^{(0)} = \infty, i = 0$, converged=False. Denote the objective value as $c^{(i)}$. }
\While{not converged} 

\State {\small $\text{Ellipsoid } \mathcal{E}_{\delta}^{(i)} \gets \textbf{P1}(\phi^{(i)}, \theta_l^{(i)}, \theta_u^{(i)})$}

\State {\small $\text{Multipliers } s_{i}(x, \theta), r_{i}(x, \theta) \gets \textbf{P2}(\phi^{(i)}, \theta_l^{(i)}, \theta_u^{(i)}, \mathcal{E}_{\delta}^{(i)})$}

\State {\small $\phi^{(i+1)},\theta_l^{(i+1)}, \theta_u^{(i+1)}, c^{(i+1)} \gets  \textbf{P3}(s_{i}(x, \theta), r_{i}(x, \theta), \mathcal{E}_{\delta}^{(i)}$)}

\If{{\small $\abs{c^{(i+1)} - c^{(i)}} \leq \epsilon$ }}
    \State {\small converged $\gets$ True}
\EndIf
\EndWhile
\end{algorithmic}
\end{algorithm}
\setlength{\textfloatsep}{2pt}

\subsection{raCBF synthesis}\label{sec:methodology:synthesis}

In this section, we provide an algorithm for synthesizing an raCBF given the partially known system (Eq.~\ref{eqn:sys_dynamics}), the safety definition (Eq.~\ref{eqn:sys_safety_def}), and a feasible initialization for the raCBF. We would like to find $\phi(x, \hat{\theta})$ satisfying Theorem~\ref{thm:valid_raCBF}, where $\phi$ is now a polynomial of fixed degree over $x, \theta$ with variable coefficients. However, with $\phi$ containing optimization variables, Eq.~(\ref{eqn:sosp_constraint_u1}),~(\ref{eqn:sosp_constraint_u2}) are no longer convex, but bilinear. This requires us to apply bilinear alternation on these constraints, which form the basis of our Alg.~\ref{alg:raCBF}. We optimize the multipliers $s_i(x, \theta), r_i(x, \theta)$ in Phase 2 and the raCBF $\phi(x,\hat{\theta})$ in Phase 3 and continue back and forth. 

\begin{align}
    & \textbf{Phase 2}\\ 
    & \min_{s_i, r_i \in \Sigma, t \in \mathbb{R}} t \nonumber\\
    &c + d \bar{u} - s_1 d - s_2 (\phi - \beta^{-}) \in \Sigma \tag*{Ctrl fltr feas. constrs} \\ 
    &c + d \ubar{u} + s_3 d - s_4 (\phi - \beta^{-}) \in \Sigma \nonumber \\
    &\phi - r_i b_i \in \Sigma \;\;,\; \forall i=1,\ldots,t \tag*{Subset constraints}\\
    &t - \phi - s_5 (\delta - x^{\top} P x) \in \Sigma \tag*{Ellipsoid constraint}\\
    &1 - \phi(0, \theta) \in \Sigma \tag*{Anchor constraint}
\end{align}

\begin{align}
    & \textbf{Phase 3}\\ 
    & \min_{\phi, t, \theta_l, \theta_u}ƒ t + \theta_l - \theta_u \nonumber \\
    &c + d \bar{u} - s_1 d - s_2 (\phi - \beta^{-}) \in \Sigma \tag*{Ctrl fltr feas. constrs} \\ 
    &c + d \ubar{u} + s_3 d - s_4 (\phi - \beta^{-}) \in \Sigma \nonumber \\
    &\phi - r_i b_i \in \Sigma \;\;,\; \forall i=1,\ldots,t \tag*{Subset constraints}\\
    &t - \phi - s_5 (\delta - x^{\top} P x) \in \Sigma \tag*{Ellipsoid constraint}\\
    &1 - \phi(0,\theta) \in \Sigma \tag*{Anchor constraint}
\end{align}

While bilinear alternation alone would yield a {\it valid} raCBF, we'd also like the raCBF to capture a large invariant set. A larger invariant set means the system is less restricted, providing the nominal controller more flexibility. Thus, our optimization needs to also maximize the size of the invariant set $\mathcal{I}^{r}_{\hat{\theta}}$. To do this, we follow~\cite{dai2022convex}. To increase its size, we perform several steps in each iteration. First, we find the largest ellipsoid $\mathcal{E}_{\delta} \defeq \{x \mid x^{\top}P x \leq \delta\}$ contained within $\mathcal{I}^{r}_{\hat{\theta}}$ (solve Phase 1 for the maximum ellipsoid radius $\delta$). Then, we optimize the raCBF so that the margin between $\mathcal{E}_{\delta}$ and $\mathcal{I}^{r}_{\hat{\theta}}$ is as large as possible (Phase 2). This effectively increases the size of $\mathcal{I}^{r}_{\hat{\theta}}$ over the iterations. We also add an ``anchor constraint'', which limits the maximum value of the raCBF. This helps us avoid a trivially optimal raCBF with large values but a small invariant set.

\begin{align}
    &\textbf{Phase 1} \\ 
    & \max_{\delta\in\mathbb{R}, s_0 \in \Sigma} r \nonumber \\
    & \text{s.t.}\;\; \delta - x^{\top} P x - s_0 \phi \in \Sigma \nonumber  
\end{align}

Another challenge posed by synthesis is the requirement for a feasible initialization. That is, we need to supply a raCBF already satisfying Eq.~(\ref{eqn:ctrl_fltr_feas_constraint}),~(\ref{eqn:subset_constraint}) when using bilinear alternation. Generally, this initialization is something that can be computed straightforwardly, but has a trivially small invariant set. Then, the idea is for synthesis to enlarge the invariant set, while maintaining validity, until it has a reasonable size. We can easily provide such an initialization if the robust range $\Theta$ is small. In that case, we can use the control Lyapunov function (CLF) corresponding to a linear quadratic regulator (LQR) which stabilizes to an equilibrium point in $\mathcal{X}_{\text{safe}}$. This initialization will intrinsically have some robustness to perturbation~\cite{chen2015robustness}. However, in the general case, the robust range $\Theta$ will have moderate size. 

To handle this, we consider a small robust range at initialization, so that we may use our LQR-based CLF initialization. Then, we maximize this robust range over the iterations (Phase 3). In our experiments, the ultimate robust range produced by our algorithm is always equal to the desired range. However, we have found, for example, that the algorithm is unable to increase the robust range to double the desired size. This is not a significant issue though, since robust control is generally not recommended when the desired range is large~\cite{lavretsky2012robustAdaptive}.

%% file: results.tex
\begin{table}[h]
\centering
\smallskip
\smallskip
\caption{}
    \begin{tabular}{|c||c|c|c|}
         \hline
         \multicolumn{4}{|c|}{ Results for Feedback Simulations}  \\
         \hline
         & Toy 2D & Cartpole & Quadrotor \\
         \hline 
         \begin{tabular}{@{}c@{}} Safety  Test \\ (\% safe trajectories) \end{tabular} & 100\% & 100\% & 100\%\\
         \hline
         \begin{tabular}{@{}c@{}}Performance  Test \\ (\% improvement of raCBF \\ over rCBF on task metric) \end{tabular} & -  & 7.014\% &  54.953\%\\
         \hline
    \end{tabular}
    \label{tab:primary_results}
\end{table}

\section{Results}\label{sec:results}

In this section, we synthesize raCBFs for three different systems: toy 2D, cartpole, and quadrotor. We show that our algorithm is quite general, producing valid raCBFs for a variety of dynamics and safety specifications. We also compare closed-loop performance in simulation for our robust-adaptive controller and a baseline robust controller~\cite{kang2023verification}. We find that our robust-adaptive safety filter ensures $100\%$ safety {\it while also} interfering less with the nominal controller. This means it achieves significantly better performance on combined safety-stabilization and safety-tracking tasks. Finally, we also show that our method scales better than the baseline, handling systems of 7D as opposed to 2D.      

In particular, we run two sets of experiments with the simulated closed-loop system:
\begin{enumerate}
    \item {\bf Safety test}: ensure the system does not leave the invariant set under an adversarial nominal controller, which tries to drive the system out. Trials vary the initial state and the true parameter(s) value(s).  
    \item {\bf Performance test}: measure stabilizing/tracking performance on combined safety-(stabilizing/tracking) task. Compare with baseline, a robust CBF~\cite{kang2023verification}. Trials vary  the true parameter(s) value(s). The cartpole's task metric is time to reach the goal and the quadrotor's task metric is the average trajectory tracking error.
\end{enumerate}
We use 200 random trials for each test. Results are listed in Table~\ref{tab:primary_results} and analyzed further in the following subsections. All the experimental hyperparameters can be found in the Appendix (Sec.~\ref{sec:appendix}). To solve the SOSPs, we use SOSTOOLS \cite{sostools} with Mosek \cite{mosek} on a 3.00GHz Intel(R) Core(TM) i9-9980XE CPU with 16 cores and 126GB of RAM. For all examples, the initial raCBF is derived from an LQR-based CLF that stabilizes the system to the origin. Specifically, we use the solution $S$ of the algebraic Riccati equation of the system linearized about the origin to define the CLF $V(x) = x^T S x$ and the raCBF $\phi^{(0)}(x) = \epsilon - V(x)$, where $\epsilon \in \mathbb{R}_{>0}$ is chosen experimentally. 


\begin{figure}[th]
    \centering
    \smallskip
    \smallskip
    \smallskip
    \includegraphics[width=0.4\textwidth]{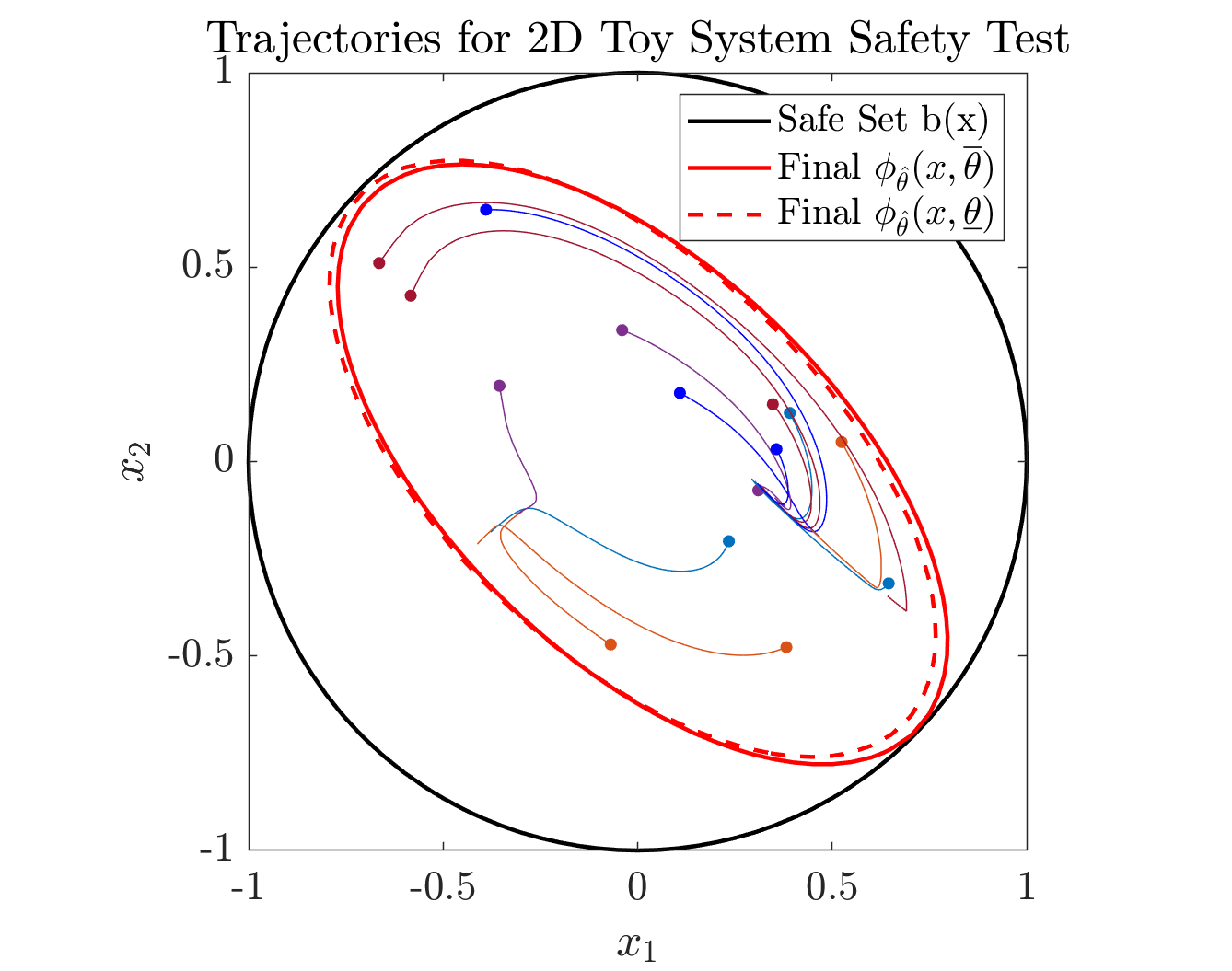}
    \caption{Random trajectories generated by the safety test. Observe that all stay within the invariant set (in red).\vspace{1em}}
    \label{fig:toy_safety_experiment}
\end{figure}

\subsection{Toy 2D}\label{sec:results:toy2d}
We use the 2D toy system from \cite{toy2dsystem}:
\begin{equation}
    \begin{aligned}
        \dot x = \begin{bmatrix} x_2 - x_1^3 \\ 0 \end{bmatrix} - \begin{bmatrix}-x_1^2 \\ 0\end{bmatrix} \theta + \begin{bmatrix} 0 \\ 1 \end{bmatrix} u 
    \end{aligned}
\end{equation}
with $\Theta \in [0.8, 1.5]$ and $u \in [-2, 2]$. The safety definition is $\mathcal{X}_{\text{safe}} = \{ x \mid b(x) \defeq 1 - (x_1^2 + x_2^2) \ge 0 \}$. 

Our synthesized raCBF attains $100\%$ safety on this system. Observe in Fig.~\ref{fig:toy_safety_experiment} that the trajectories under the safety-filtered adversarial controller never leave the system despite the parametric uncertainty. 

\begin{figure}[htbp]
    \centering
    \includegraphics[width=0.42\textwidth]{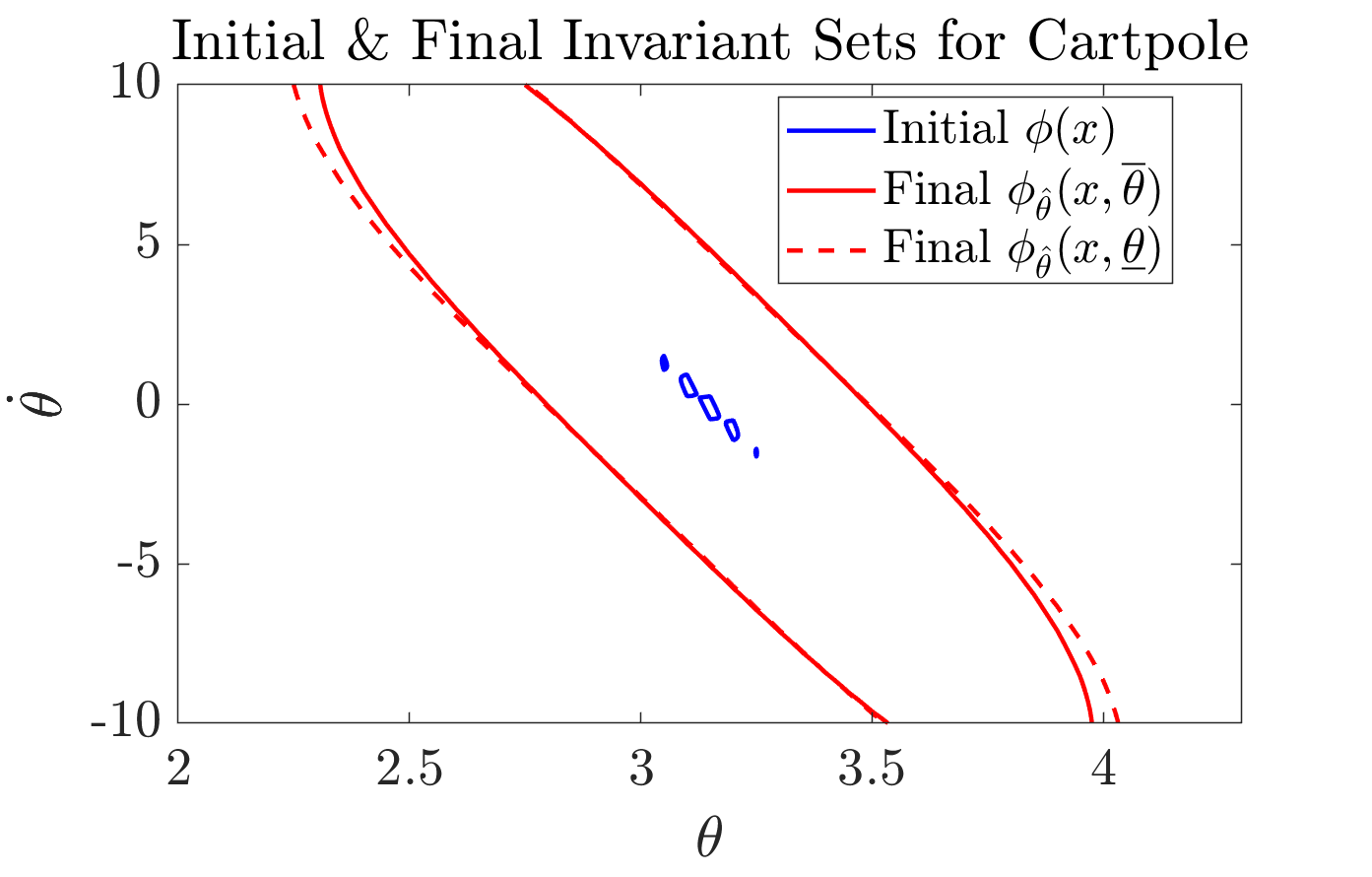}
    \caption{Note the significant increase in size from the initial to final invariant sets. Also, observe that the final raCBF depends on $\hat{\theta}$ in an intuitive way.\vspace{1em}}
    \label{fig:cartpole_invariant_sets}
\end{figure}

\subsection{Cartpole}\label{sec:results:cartpole}

For the cartpole system (Fig.~\ref{fig:systems_side_by_side}), we transform the standard trigonometric dynamics~\cite{cartpole} to polynomial form by redefining the states as $[\dot x, x, {\small \dot\psi}, \sin(\psi), \cos(\psi) +1]$. For these new states, we require $\sin^2(\psi) + \cos^2(\psi) = 1$, which we incorporate into the SOSP using the S-procedure~\cite{posa2016trigsub}. The unknown parameter is the friction coefficient at the joint. We let $\Theta \in [0.28, 0.31]$, following \cite{barto1983neuron} and $u \in [-5, 5]$ N. The safety definition confines the pole angle to be small: $\mathcal{X}_{\text{safe}} \defeq \{\sin^2(\frac{\pi}{3}) - \sin^2(\psi) \ge 0 \}$. 


In Fig.~\ref{fig:cartpole_invariant_sets}, we plot the invariant set along $(\theta, \dot{\theta})$. First, observe that the shape of the set is sensible: it is an ellipsoid with primary axis $\dot{\theta} = -\theta$, meaning it considers the pole swinging to the vertical as safe and vice versa. Secondly, note that the synthesized raCBF depends meaningfully on $\hat{\theta}$. Comparing the invariant sets for $\phi(x, \ubar{\theta})$ (low joint friction) and $\phi(x, \bar{\theta})$ (high joint friction), we can see that the former set is larger. This makes sense, since for low joint friction, the pole is more responsive and easily controlled. Finally, notice that the initial invariant is almost negligibly small, but the final set is quite large. We conclude that our algorithm is very effective at growing the invariant set. 

As for the closed-loop simulations, raCBF attains $100\%$ safety, even with an adversarially unsafe nominal controller. We also synthesize a robust CBF according to~\cite{kang2023verification}. In the performance test, we implement a nominal controller which controls the cart to reach position $x_{goal}$. Then, the nominal controller is either layered with an raCBF-based or an rCBF-based safety filter, both of which prevent the pole from tipping.  

Overall, the raCBF controller reaches the goal $7.014\%$ faster. One reason for this is that the raCBF has a larger invariant set (Fig.~\ref{fig:raCBF_vs_rCBF_invariant_sets}), which allows larger pole angles during safe operation (Fig.~\ref{fig:cart_reach_trajectories_compare}) and hence larger cart acceleration. The difference in invariant set size can be explained as follows: while both raCBF and rCBF are robust to maximum estimation error, the raCBF is able to account for this error in a more moderate way, since it assumes parameter estimation is possible. 

The second reason raCBF reaches the goal faster is its use of parameter estimation. The parameter adaptation law renders the raCBF {\it selectively conservative}, whereas the rCBF is always conservative (always operating under worst-case assumptions). To see this, observe Fig.~\ref{fig:cart_reach_theta_psi}. The law adapts the parameters toward values that increase conservatism when the state becomes increasingly unsafe. This happens at times $\approx 3, 7, 9$ seconds, when the pole changes direction as the cart oscillates around the goal. This pivot causes the parameter estimate to grow to $\bar{\theta}$, which corresponds to the smallest and thus most conservative invariant set ($\mathcal{I}^{r}_{\bar{\theta}}$) among the family of invariant sets $\{ \mathcal{I}^{r}_{\hat{\theta}}, \hat{\theta} \in \Theta \}$. On the other hand, when the state is relatively safe, minimal conservatism is applied. This is the case when the pole is static. At those points, the parameter estimate is  $\ubar{\theta} = 0$, corresponding to the largest invariant set in the family of invariant sets.

\begin{figure}[htbp]
    \centering
    \smallskip
    \smallskip
    \includegraphics[width=0.35\textwidth]{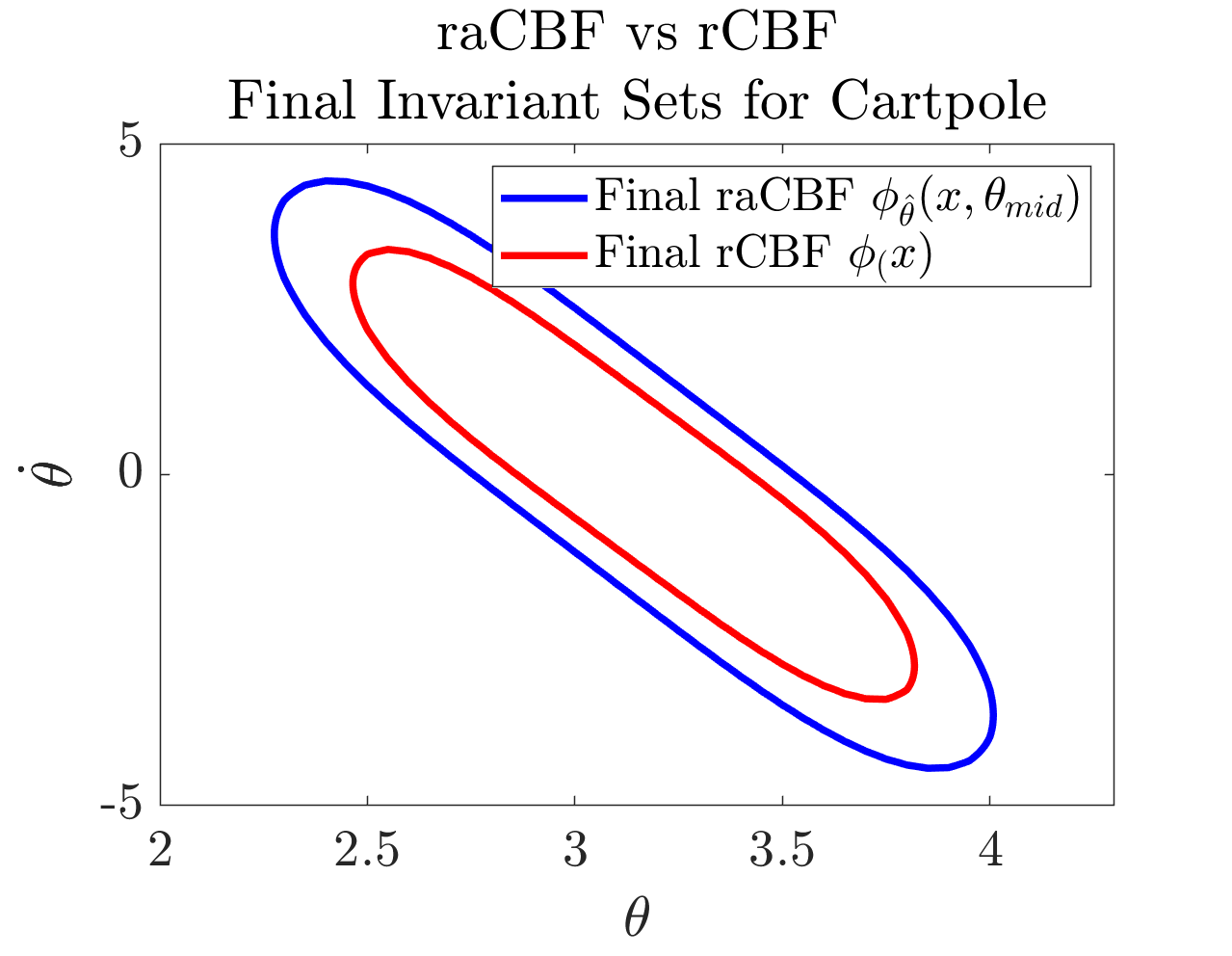}
    \caption{Notice that raCBF's invariant set is much larger than the baseline's, which accounts for the difference in the performance test.}\label{fig:raCBF_vs_rCBF_invariant_sets}
\end{figure}

\begin{figure}[htbp]
    \centering
    \includegraphics[width=0.45\textwidth]{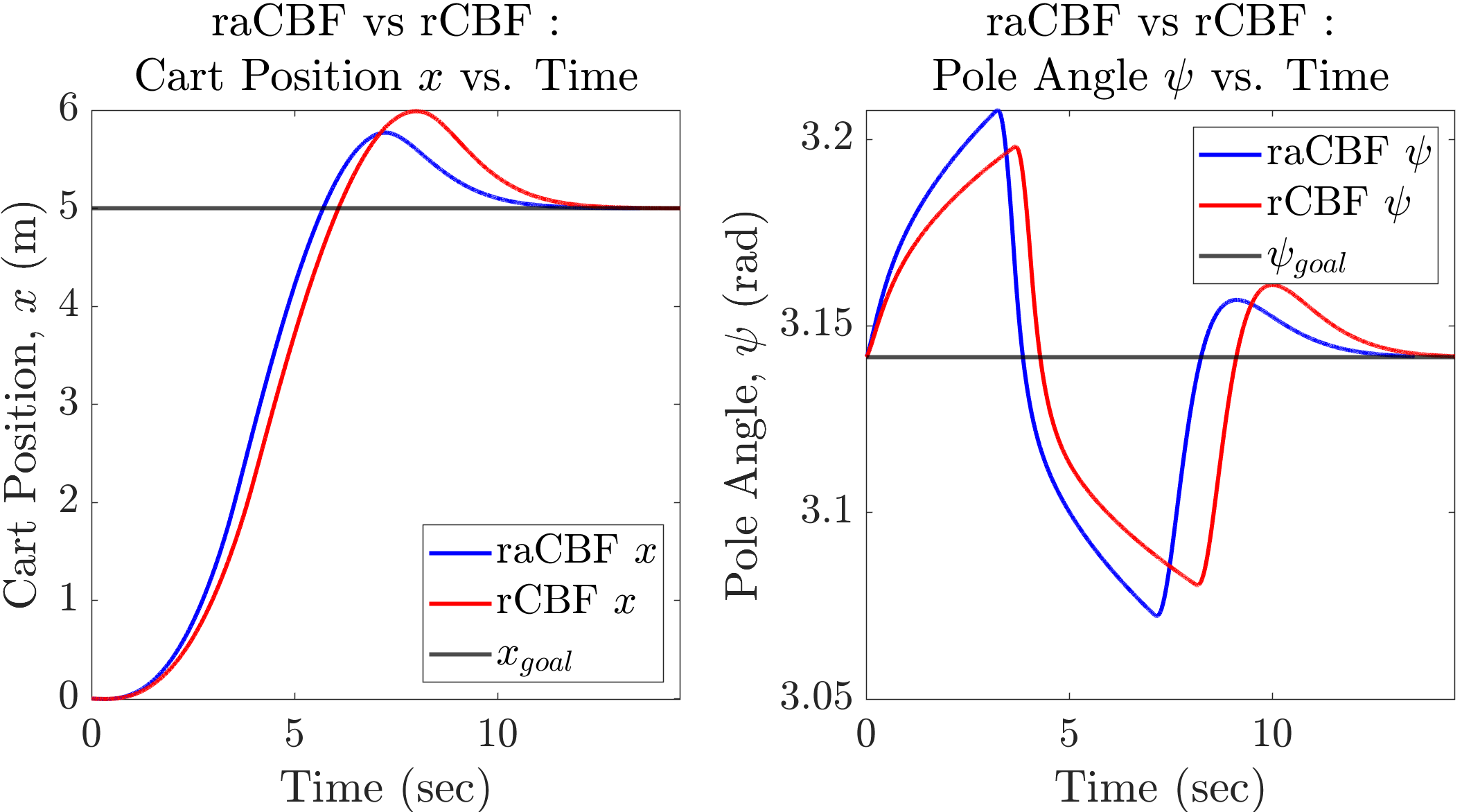}
    \caption{Detailed analysis of one trial from cartpole's performance test.\vspace{1em}}
    \label{fig:cart_reach_trajectories_compare}
\end{figure}

\begin{figure}[htbp]
    \centering
    \includegraphics[width=0.35\textwidth]{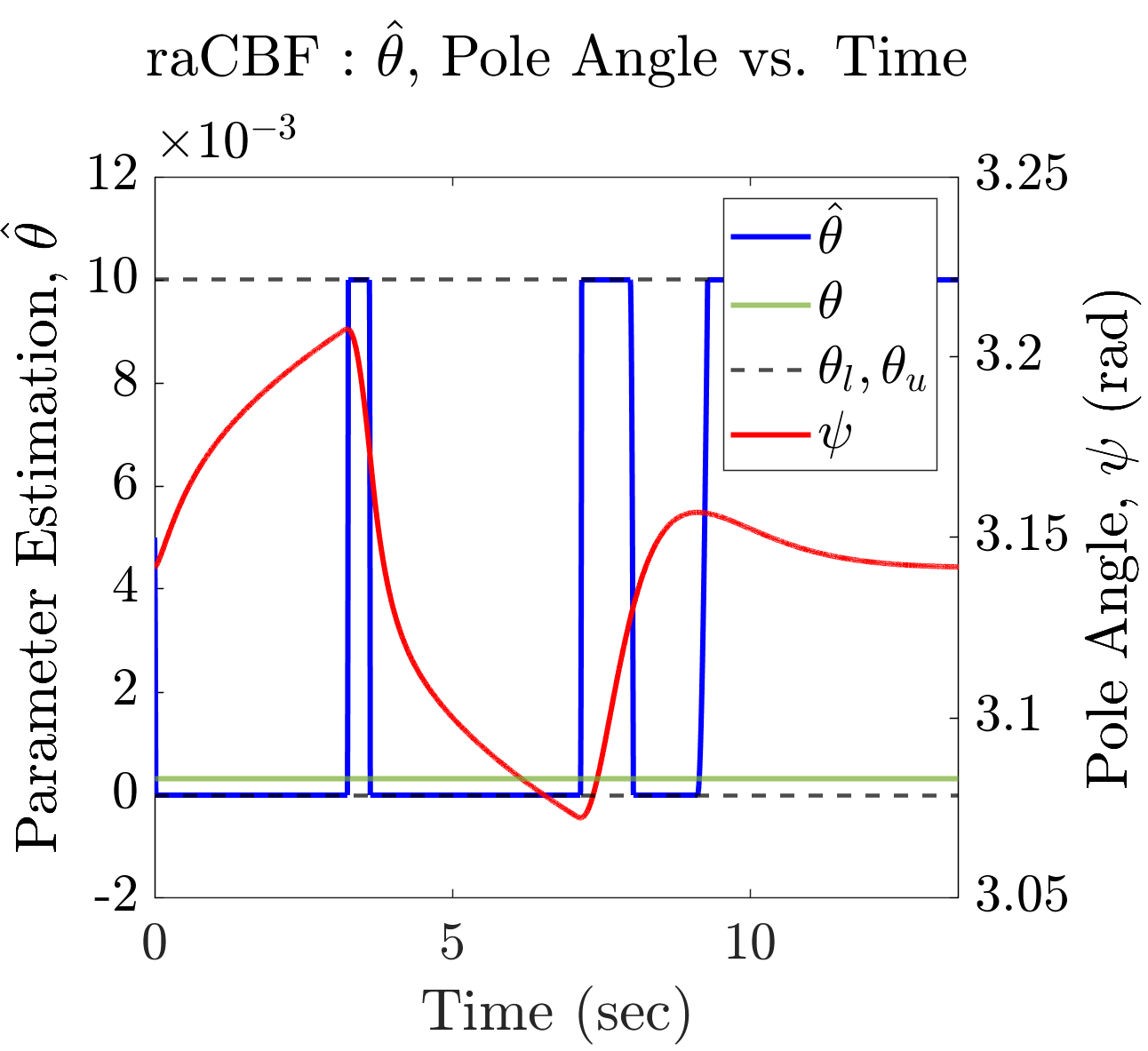}
    \caption{Analysis of parameter estimation for cartpole's performance test.}
    \label{fig:cart_reach_theta_psi}
\end{figure}

\subsection{Quadrotor}\label{sec:results:quadrotor}
This system is a planar quadrotor (Fig.~\ref{fig:systems_side_by_side}) that we transform from the typical trigonometric dynamics~\cite{luuk2011quad2d} redefining the states as $[\dot x, \dot y, \dot \psi, x, y, \sin(\psi), \cos(\psi) - 1]$. We handle the transformation in much the same way as the cartpole (Sec.~\ref{sec:results:cartpole}). The system has two unknown parameters, which are the drag coefficients in $x$ and $y$ that range from $[0, 0.01]$, roughly following the physical values from~\cite{ventura2018high}. The propeller thrusts are limited to $[-4, 4]$ mg. The safety definition enforces the system to avoid walls on the left and right: $\mathcal{X}_{safe} = \{x \mid x \in [-1, 1] \}$.    

Our synthesis is able to handle a system of this moderate size (7D), producing an raCBF that ensures safety in $100\%$ of the randomized trials of the safety test (Table~\ref{tab:primary_results}). The raCBF also has a large invariant set, which is evidenced by how it can track trajectories closely to the wall (Fig.~\ref{fig:quad2d_track_raCBF_vs_rCBF}). On the other hand, the robust CBF from~\cite{kang2023verification} cannot be applied to a system of this size. However, in order to produce some point of comparison, we transfer the key parts of their formulation into our SOSP algorithm. Essentially, this just modifies the control filter feasibility constraints in phases 2 and 3 of our algorithm. This allows us to produce a result after six hours of optimization, which we use in the following performance test.

In the performance test, both controllers attempt to track a trajectory that crosses the walls. The nominal tracking controller is gain-scheduled LQR. As we can see, our raCBF offers an $54.95\%$ improvement over the robust controller. In Fig.~\ref{fig:quad2d_track_raCBF_vs_rCBF}, we note that while our raCBF allows the quadrotor to approach the wall, the robust CBF, which has a much smaller invariant set in the $x$ dimension, restricts the quadrotor to around $x=0$. We can conclude that raCBF performs better and currently scales better than rCBF. More work will be required to get rCBF to scale moderately. 

\begin{figure}[tp]
    \centering
    \smallskip
    \includegraphics[width=0.45\textwidth]{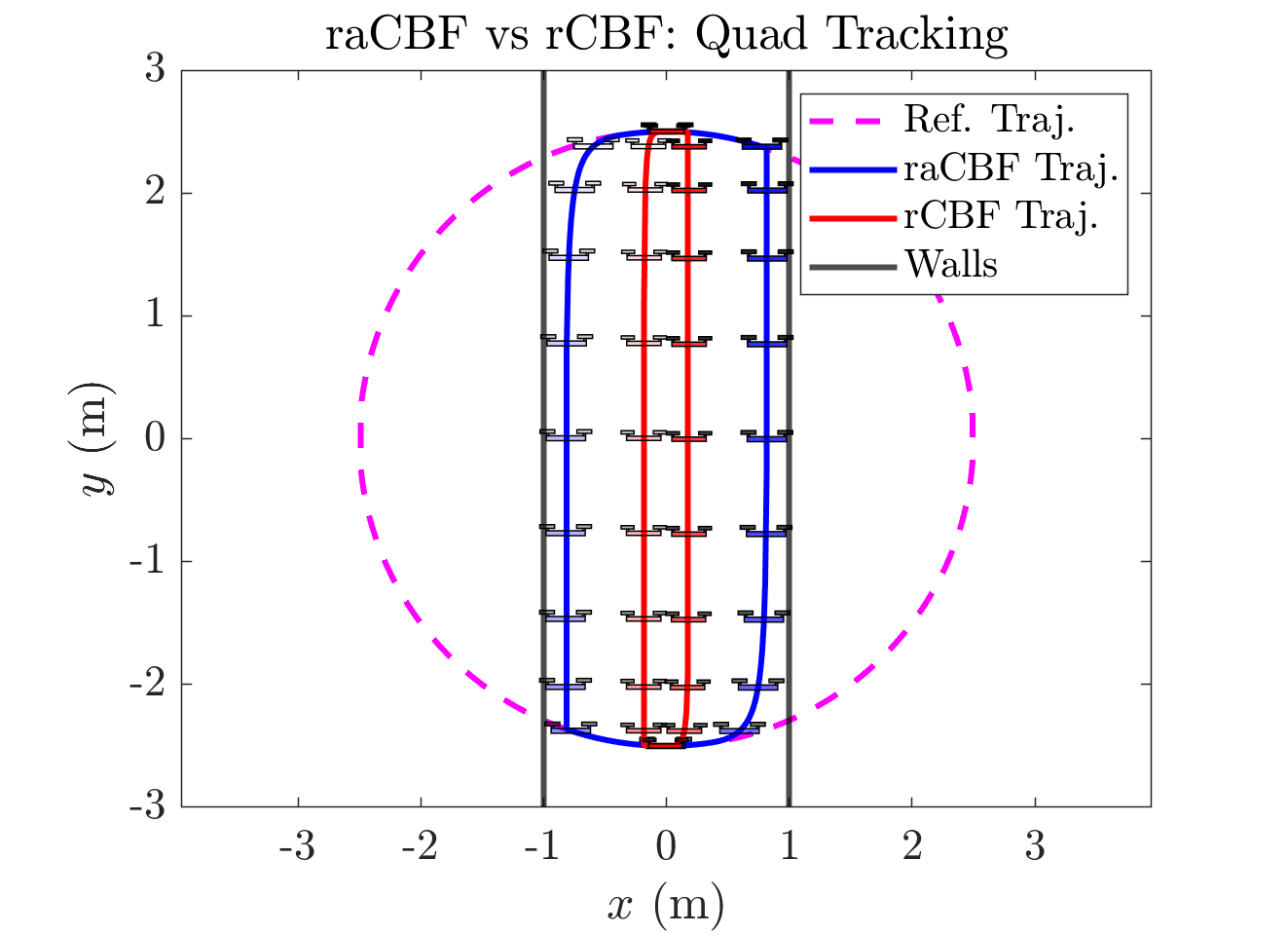}
    \caption{raCBF enables closer trajectory tracking than rCBF and is therefore more performant.\vspace{1em}}
    \label{fig:quad2d_track_raCBF_vs_rCBF}
\end{figure}

%% file: conclusion.tex
\section{Conclusion}\label{sec:conclusion}

In this paper, we proposed the first verification and synthesis methods for adaptive-type safe controllers for uncertain systems. We proved that the constraints characterizing a valid raCBF can be written as an SOS program, and showed we can devise a multi-phase algorithm to synthesize an raCBF. Then, in our experiments, we illustrated that our algorithm is generic and scalable with three examples with varied dynamics, dimensions (up to 7D), and safety specifications. We empirically confirmed our theoretical safety guarantees by showing that $100\%$ of trajectories under the safe controller are safe. Finally, we showed that our raCBF provides up to $\approx 55\%$ task performance improvement over the baseline. In the future, we are interested in extending synthesis to systems with time-varying unknown parameters. 

%% file: appendix.tex

~
\begin{table}[h!]
\centering
    \caption{}
    \label{tab:comp_time}
    \begin{tabular}{|c||c|c|c|}
         \hline
         \multicolumn{4}{|c|}{Computation Lengths for Synthesis}  \\
         \hline
          & 2D Toy & Cartpole & Quadrotor  \\
         \hline
         Total Time (sec) & 23 & 655 & 68,934\\
         \hline
         Total Iterations & 11 & 159 & 337\\
         \hline
         Avg. P1 Time (sec) & 0.1 & 0.12 & 0.87\\
         \hline
         Avg. P2 Time (sec) & 1 & 2.8 & 74\\
         \hline
         Avg. P3 Time (sec) & 0.8 & 1.2 & 130\\
         \hline
    \end{tabular}
\end{table}


\begin{table}[h!]
\centering
    \caption{}
    \label{tab:experiment_details}
    \begin{tabular}{|c||c|c|c|}
         \hline
         \multicolumn{4}{|c|}{Synthesis \& Estimation Hyperparameters}  \\
         \hline
          & 2D Toy & Cartpole & Quadrotor  \\
         \hline
         Deg. $\phi(x,\hat{\theta})$ & 3 & 3 & 3\\
         \hline
         $\beta^{-}$ & -1.0 & -0.5  & -0.1\\
         \hline
         $\alpha$ & 0.1 & 0.01 & 0.01\\
         \hline
         $\gamma$ & 10.0 & 10.0 & 10.0\\
         \hline
         $\nu(\rho)$ & $\text{atan}(\rho)+1$ & $\text{atan}(\rho)+1$ & $\text{atan}(\rho)+1$ \\
         \hline
         $\eta$ & 100 & 100 & 1000 \\
         \hline
         $\sigma$ & 1 & 1 & 1 \\
         \hline
         $\xi$ & 1 & 1 & 1 \\
         \hline
         $\rho$ range & [0,10] & [0,10] & [0,10] \\
         \hline
         
    \end{tabular}
\end{table}


\subsection{System Parameters}

{\small
\noindent $\bullet$ Cartpole: $m_{c}=1$kg, $m_{p}=0.1$kg, $l=0.5$m \\
\noindent $\bullet$ Quadrotor: $m=0.486$kg, $r=0.25$m, $I=0.00383$kg$\cdot$m$^2$ 

\subsection{Final raCBF}
\noindent {\bf 2D toy:} $\phi(x,\hat{\theta}) = (0.3666)\hat{\theta}^2 + (0.07823)\hat{\theta}x_1 + (0.0007759)\hat{\theta}x_2 \\- (1.673)x_1^2 - (2.005)x_1x_2 - (1.726)x_2^2 - (0.8380)\hat{\theta} - (0.09071)x_1 \\- (0.0007912)x_2 + 1.098$\\
\\
\noindent {\bf Cartpole:} $\phi(x,\hat{\theta}) =  (23.24)\hat{\theta}^2 + (0.0002120)\hat{\theta}z_1 - (22.73)\hat{\theta}z_2 \\- (0.04065)x_3^2 + (1.136)x_3z_1 - (10.06)z_1^2  - (8.015)z_2^2  - (12.98)\hat{\theta} \\+ (10.27)z_2 + 2.785$ \\
\\
\noindent {\bf Quadrotor:} $\phi(x,\hat{\theta}) = (-2.663)\hat{\theta}_1^2 - (0.06829)\hat{\theta}_1\hat{\theta}_2 \\+ (0.05703)\hat{\theta}_1z_2 + (0.3199)\hat{\theta}_2^2 - (0.08092)\hat{\theta}_2z_2 - (0.7239)x_1^2 \\+ (0.03167)x_1x_3 - (1.138)x_1x_4 + (3.074)x_1z_1  - (0.01136)x_3^2 \\+ (0.02167)x_3x_4 - (0.1209)x_3z_1  - (1.462)x_4^2 + (2.118)x_4z_1 \\- (5.706)z_1^2 - (5.308)z_2^2 + (0.06482)z_2 + 0.9881$
}
    
    